\documentclass[conference]{IEEEtran}

\usepackage[utf8]{inputenc}
\IEEEoverridecommandlockouts
\usepackage{cite}
\usepackage{enumitem}
\usepackage{amsmath,amssymb,amsfonts}
\usepackage{algorithmic}
\usepackage{graphicx}
\usepackage{textcomp}
\usepackage[dvipsnames]{xcolor}
\usepackage{url}
\usepackage{tcolorbox}
\usepackage{subcaption}
\usepackage{balance}
\def\BibTeX{{\rm B\kern-.05em{\sc i\kern-.025em b}\kern-.08em
    T\kern-.1667em\lower.7ex\hbox{E}\kern-.125emX}}

\usepackage[disable]{todonotes}    
\newcommand{\sketch}{\todo[color=yellow!30,inline]}    
\usepackage[linesnumbered,vlined,boxed,commentsnumbered, ruled]{algorithm2e} 
\usepackage{amsthm}
\definecolor{light-gray}{gray}{0.95}

\usepackage{pgfplots}
\pgfplotsset{width=\columnwidth,compat=1.9}
\usepgfplotslibrary{external}
\usetikzlibrary{patterns}

\newtheorem{theorem}{Theorem}

\usepackage[absolute,showboxes]{textpos}
\TPMargin{5pt}
\textblockcolour{gray!02}

\newcommand{\copyrightstatement}{
	\begin{textblock}{1}(0.00,0.0)    
		\noindent
		\scriptsize
		\copyright \  
		2021 IEEE.
		Personal use of this material is permitted.  Permission from IEEE must be obtained for all other uses, in any current or future media, including reprinting/republishing this material for advertising or promotional purposes, creating new collective works, for resale or redistribution to servers or lists, or reuse of any copyrighted component of this work in other works.
		This paper is the authors' accepted version to be published in the 40th IEEE Symposium on Reliable Distributed Systems (SRDS). For the final, published version we refer to DOI [\textit{to be inserted here later upon publication}].
	\end{textblock}
}

\begin{document}
	\copyrightstatement


\title{Making Reads in BFT State Machine Replication Fast, Linearizable, and Live
\thanks{This work has been funded by the Deutsche Forschungsgemeinschaft (DFG, German Research Foundation) grant number 446811880 (BFT2Chain), by FCT through the ThreatAdapt project (FCT-FNR/0002/2018), and the LASIGE Research Unit
(UIDB/00408/2020 and UIDP/00408/2020), and by the European Commission
through the H2020 VEDLIoT project (957197).}
}


 \author{\IEEEauthorblockN{Christian Berger}
  	\IEEEauthorblockA{
  		\textit{Universität Passau}\\
  		Passau, Germany 
  	}
  	\and
  	\IEEEauthorblockN{Hans P. Reiser}
  	\IEEEauthorblockA{
  		\textit{Universität Passau}\\
  		Passau, Germany 
  	}
  	\and
  \IEEEauthorblockN{Alysson Bessani}
  \IEEEauthorblockA{
  \textit{LASIGE, Faculdade de Ciências,} \\
 \textit{Universidade de Lisboa},  
   Portugal 
  }
 }

\maketitle

\begin{abstract}
\textit{Practical Byzantine Fault Tolerance} (PBFT) is a seminal state machine replication protocol that achieves a performance comparable to non-replicated systems in realistic environments.
A reason for such high performance is the set of optimizations introduced in the protocol.
One of these optimizations is \textit{read-only requests}, a particular type of client request which avoids running the three-step agreement protocol and allows replicas to respond directly, thus \textit{reducing the latency of reads from five to two communication steps}.
Given PBFT's broad influence, its design and optimizations influenced many BFT protocols and systems that followed, e.g., BFT-SMaRt.
We show, for the first time, that the read-only request optimization introduced in PBFT more than 20 years ago \textit{can violate its liveness}.
Notably, the problem affects not only the optimized read-only operations but also standard, totally-ordered operations.
We show this weakness by presenting an attack in which a malicious leader blocks correct clients and present two solutions for patching the protocol, making read-only operations fast and correct.
The two solutions were implemented on BFT-SMaRt and evaluated in different scenarios, showing their effectiveness in preventing the identified attack.
\end{abstract}

\begin{IEEEkeywords}
Byzantine Fault Tolerance, State Machine Replication, Reads, Attack
\end{IEEEkeywords}

\section{Introduction}\label{intro}

A standard mechanism for implementing dependable services is \emph{state machine replication} (SMR): 
it achieves fault tolerance by coordinating client interactions with a set of server replicas, ensuring all replicas process the same requests in the same order~\cite{schneider1990implementing}.
This way, the service remains functional even if a fraction of the replicas fail.
A fault model describes the type of faults that are considered.
For example, \emph{Byzantine fault tolerance} (BFT) allows replicas to display arbitrary behavior, even involving collusion with other faulty replicas~\cite{lamport1982byzantine}.
A Byzantine replica may selectively send messages or send conflicting messages to other replicas.
In BFT systems, an adversary should not be able to break the system as long as it controls no more than $f$ out of $n$ replicas. 

A major breakthrough has been made by the introduction of PBFT~\cite{castro1999practical}:
it describes a practical BFT SMR protocol that performs well in realistic environments by incorporating a set of optimizations to increase the performance of the BFT system.
It claims that all optimizations introduced preserve the liveness and safety properties of the protocol~\cite{castro1999practical}.

Design choices of PBFT have been widely adopted by other BFT protocols that made further enhancements, e.g., HQ~\cite{cowling2006hq}, Zyzzyva~\cite{kotla2007zyzzyva}, PBFT-CS~\cite{wester2009tolerating}, Aardvark~\cite{clement2009upright}, and UpRight~\cite{clement2009upright}.
Also, the more recent BFT-SMaRt~\cite{bessani2014state} uses a similar normal-case operation message pattern and inherits several optimizations, among them the tentative use of unordered  \textit{read-only requests}.
This optimization allows a client to read the state of the service in a single round-trip, without passing through the ordering protocol, as long as it waits for larger quorums of matching replies.
Incorporating optimizations like this can lead to substantial performance gains~\cite{sousa2015separating} but can also make the system deviate from the original protocol, which was proved correct, thus opening the door for vulnerabilities that subsequently manifest in a highly optimized BFT system's implementation and can be exploited by attackers.


\paragraph*{Contributions}
In this paper, we identify a problem in a widely-accepted 20-year-old optimization of BFT SMR design that enables fast read operations~\cite{castro1999practical}.
The attack we present works against the normal-case operation pattern of a highly optimized BFT system.
In particular, this attack is possible if PBFT's \textit{read-only operations} optimization is used and the adversary controls $f$ replicas, including the leader.
The attack is not limited to (unordered) read-only operations but affects standard totally-ordered operations (e.g., state updates) alike.
Such a leader can block selected clients forever, thus negatively affecting the liveness of the system.

The identified weakness stems from the requirement that clients wait for strictly more than $n-2f$ matching replies to preserve SMR consistency.
Therefore, we believe the same vulnerability might appear in other PBFT-like protocols requiring clients to wait for such quorums~\cite{Bessani13,Veronese13}.
However, in this paper, we focus our attention on the problem of fast, consistent, and live reads.

We present two generic solutions to patch a vulnerable PBFT-like SMR protocol with the read-only optimization. 
The solutions enable the system to guarantee both liveness (assuming partial synchrony) and safety (i.e., linearizability~\cite{herlihy1990linearizability}) even when single round-trip reads are enabled.
The core idea of our solutions is to modify the normal case operation to ensure requests are eventually executed in all correct replicas, or at least all these replicas are able to reply to pending requests.
Note that the SMR liveness is typically defined by clients eventually being able to accept a valid response to a request~\cite{castro1999practical}, which does not necessarily require all correct replicas to commit and execute a request.

A further novel observation we make is that, if the read-only optimization is used, the state transfer needs to make sure recovered replicas can reply to pending requests, thus making transferring replies during a state transfer necessary.

\paragraph*{Outline}
In the remainder of this paper, we first present relevant background in Section~\ref{background} and then elaborate on how reads are optimized in SMR-based systems
in Section~\ref{optimizing-reads}.
Subsequently, in Section~\ref{problem}, we explain the problem by presenting an attack on PBFT and BFT-SMaRt that can break the liveness of the system if the read-only optimization is used. 
After that, in Section~\ref{solution}, we explain our solutions to make reads in BFT SMR fast, linearizable, and live.
In Section~\ref{evaluation} we evaluate our implemented solutions in a LAN and a WAN environment. Finally, we discuss related work in Section~\ref{related-work} and conclude in Section~\ref{conclusion}.

\section{Background and Preliminaries}
\label{background}

We first explain our system model (Section~\ref{system-model}), some general background on BFT SMR (Section~\ref{bft-smr}), and subsequently review the PBFT algorithm (Section~\ref{pbft}).

\subsection{System Model} \label{system-model}
We consider a distributed system with a set of clients $C$ and a total of $n = 3f+1$ independent server replicas $R = \{r_0, r_1, ..., r_{n-1}\}$ where $f$ is an upper bound on the number of faulty replicas present in this system.
We assume an adversary that has full control over the $f$ Byzantine replicas that can be freely selected from $R$.
Although the behavior of faulty replicas is arbitrary, it is still bound by their computational capabilities, e.g., Byzantine replicas cannot break cryptographic primitives.

Every pair of processes communicates through authenticated fair point-to-point links, i.e., messages can be lost and delayed, but not forever.
Further, we assume a partially synchronous~\cite{dwork1988consensus} model in which processing and network can be initially asynchronous but after an unknown global stabilization time \textit{GST}, the communication and computation delay becomes bounded by some unknown values. 

\subsection{State Machine Replication}
\label{bft-smr}

State machine replication (SMR) is a classical paradigm for implementing fault-tolerant services in which a set of replicas emulates a centralized service processing operations.
The implementation of the SMR paradigm requires three properties~\cite{schneider1990implementing}: (1- Initial State) all correct replicas start on the same state; \mbox{(2- Determinism)} all correct replicas receiving the same operation on the same state produce the same result and resulting state; and (3- Coordination) all correct replicas process the same operations in the same order.
The third requirement, in particular, typically requires a consensus protocol for ensuring all correct replicas process the same operations in the same order.

Assuming these three requirements are satisfied by a service, the SMR approach satisfies the following properties:

\begin{itemize}
\item \emph{Safety}: all correct replicas execute the same sequence of operations;
\item \emph{Liveness}: all operations issued by correct clients are eventually completed.
\end{itemize}

The liveness property deserves further explanation.
An operation is considered completed when the client receives enough matching replies.
So, what needs to be ensured is that every request from a correct client is ordered and delivered to \emph{enough} correct replicas that execute it and reply to the client.
With crash failures, a single reply is enough for the client to know its operation was executed.
With Byzantine failures, as in this paper, a client operation is considered completed when $f+1$ matching replies from different replicas are received.

The safety property of SMR aims to ensure  \textit{linearizability}~\cite{herlihy1990linearizability} of the replicated service.
Roughly speaking, it states that the replicated state machine should \textit{behave} exactly like a centralized implementation (thus providing the illusion as if interacting with a single correct state machine, which atomically executes operations one after another).
In particular, the effect of an update operation that completes should be reflected in all subsequent (update or read) operations.

\subsection{The PBFT Algorithm} \label{pbft}

The Practical Byzantine Fault Tolerance (PBFT)~\cite{castro1999practical} algorithm was arguably the first \textit{practical} approach for tolerating Byzantine faults.
It guarantees safety under asynchrony and liveness in a partially synchronous system model~\cite{dwork1988consensus}, as long as less than a third of replicas are faulty.
PBFT's practicality stems from such optimal resilience and its achieved performance, which is comparable to non-replicated systems~\cite{castro1999practical}.

\begin{figure}[t]
    \centering
    \includegraphics[width=1\linewidth]{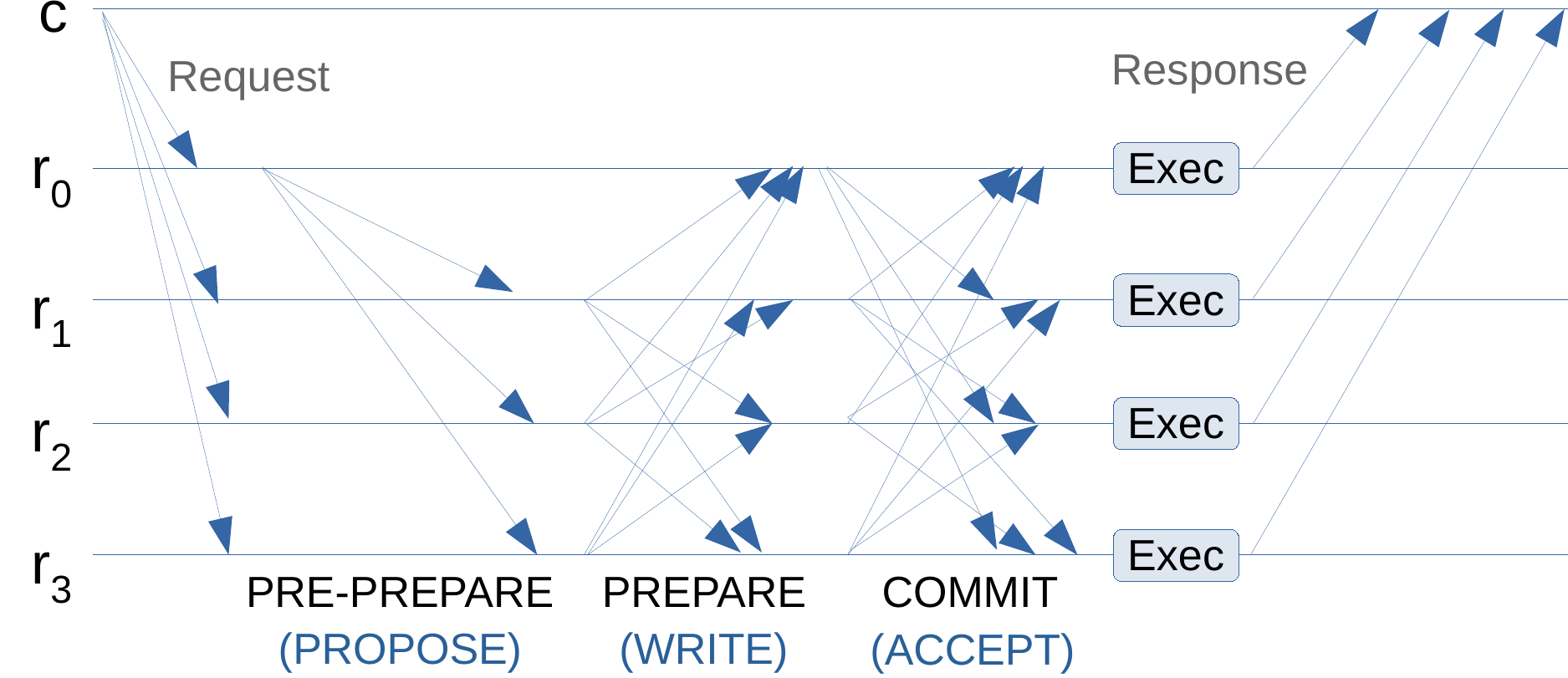}
    \caption{Normal case operation of PBFT (and BFT-SMaRt).}
    \label{fig:bft}
\end{figure}

\begin{figure*}[!t] 
 \centering
  \includegraphics[width=0.95\textwidth]{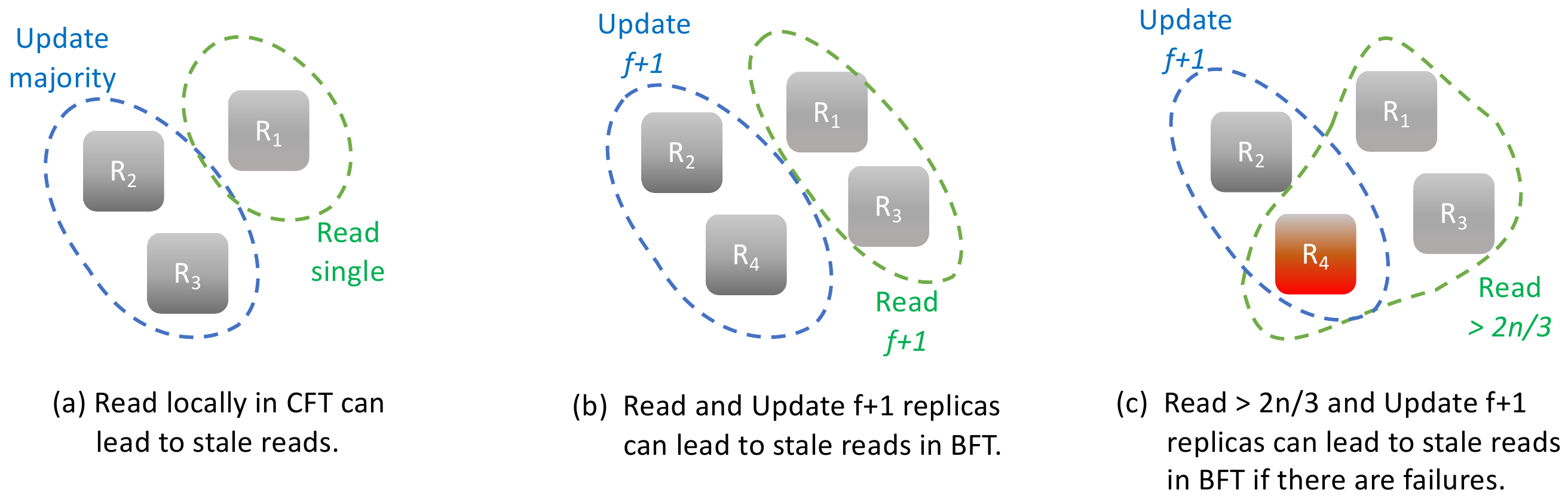}
  \caption{Quorum intersection of read-only optimizations for SMR.}
  \label{fig:quorums} 
\end{figure*}

PBFT's message flow is illustrated in Figure~\ref{fig:bft}.
Operations start when a client sends its request to the replicas (the leader, or after a timeout to all). 
Subsequently, the leader creates a batch of client requests and broadcasts it in a \texttt{PRE-PREPARE} message to all other replicas, which need to decide on it.
Replicas use a two-phase commit pattern of \texttt{PREPARE} and \texttt{COMMIT} which are all-to-all broadcasts: in both of these steps, each replica $r_i$ waits for a quorum certificate containing a total of $2f+1$ replicas (including itself) necessary to make a decision. All quorums are sufficiently large to overlap in at least one correct replica thus ensuring agreement among all, e.g., for any replica $r_i$ and $r_j$ it holds for their quorums $Q_i$ and $Q_j$ that $\vert Q_i \cap Q_ j  \vert \geq f+1 $. This ensures the safety of the protocol. Since commitment is ensured in at least $2f+1$ replicas, at least $f+1$ \textit{correct} replicas execute and respond to the client. 
The client accepts a result as soon as it gathers a weak certificate of $f+1$ \textit{matching} responses (albeit actual results may be replaced by a digest here).

If the leader is correct and no more than $f$ replicas are faulty, then the above described \textit{normal case operation} succeeds. 
The problem of a faulty leader can be resolved by the \textit{view change} protocol if enough ($f+1$) replicas support a view change.




\section{Optimizing Reads in SMR}
\label{optimizing-reads}

Many practical systems tend to execute much more reads than updates (e.g., \cite{hunt2010zookeeper,corbett2013spanner}).
Intuitively, the execution of consensus seems to be too costly when a client just wants to read (a portion of) the state of the replicated service.
This is particularly true when considering wide-area deployments, where the time required for running a five-step protocol like PBFT can be significant.

For this reason, crash fault-tolerant (CFT) SMR-based systems tend to implement methods for reading from a single replica.
The main issue with reading from a single replica in a naive way is that it negatively affects the consistency guarantee of the system.
The reason for this is that in a non-synchronous system model, an operation (say, an update) is considered committed and ready for execution when it is confirmed by a quorum of replicas, i.e., a subset of replicas of size $q$, being $q \leq n-f$ for availability~\cite{malkhi1998byzantine}.
With such a requirement, a local read happening after an update is executed, as shown in Figure~\ref{fig:quorums}(a), may not reflect the committed update.

This is what happens, for instance, in Zookeeper~\cite{hunt2010zookeeper}, which allows a client to read from any replica.
However, this optimization, when used with FIFO communication channels, makes the system satisfy only sequential consistency~\cite{attiya1994seqconsvslin}, and not linearizability, i.e., clients can read stale values, but these are part of a consistent sequential history of operations~\cite{Jun11}.
Satisfying only sequential consistency enables Zookeeper to scale its capacity to serve reads with its number of replicas.

To preserve linearizability and still read from a single replica, some CFT systems employ temporary leases.
The most common type of lease is to read only from the stable Paxos leader~\cite{corbett2013spanner}, but there are also more advanced ways to implement such leases for giving the right to serve reads on different parts of the state to different replicas~\cite{moraru2014paxos}.

When tolerating Byzantine failures, it is clearly impossible to read from a single replica, since the result might not reflect the correct state of the system.
Therefore, the minimum number of replicas to be contacted will be $f+1$.
However, reading from this amount of replicas still leads to linearizability violations, as illustrated in Figure~\ref{fig:quorums}(b).

To the best of our knowledge, the only way to perform a read without executing consensus or giving up linearizability is by employing the read-only optimization introduced in PBFT~\cite{castro1999practical}.
Read-only requests are handled by a direct request-reply pattern, in which all replicas optimistically respond to a read without ordering the operation.
This requires not only the read operation, but all (update) operations to wait for a quorum of $q > 2n/3$ matching responses at the client to ensure linearizability.
The reason for this requirement is illustrated in Figure~\ref{fig:quorums}(c): the intersection between a quorum $Q_\mathit{read}$ with $q > 2n/3$ replicas and a quorum $Q_\mathit{update}$ with only $f+1$ replicas can contain only malicious replicas that will answer an outdated value.
This requirement is also explicitly stated by Castro and Liskov:
``\textit{The read-only
optimization preserves linearizability provided clients obtain $2f+1$ matching replies for both read-only and read-write operations}.''~\cite{castro2001byzantine}.
If 
the client cannot obtain $q$ matching replies, it re-transmits the read-only operation as a normal, ordered operation.

\section{The Problem in PBFT and BFT-SMaRt} \label{problem}

The reason for the optimization discussed in the previous section to work is because it implicitly assumes that eventually all \emph{correct replicas} will execute all updates submitted to the system.
This comes from the Agreement property of total order (or atomic) multicast~\cite{Had94}: \emph{if a correct process delivers $m$, then all correct processes eventually deliver $m$.}\footnote{Or, alternatively, due to the consensus' Termination property~\cite{Had94}: \emph{all correct processes eventually decide.}}
In this section, we show this is not the case due to subtle issues affecting both PBFT and BFT-SMaRt.

\begin{figure}[htb]
    \centering
    \includegraphics[width=1\linewidth]{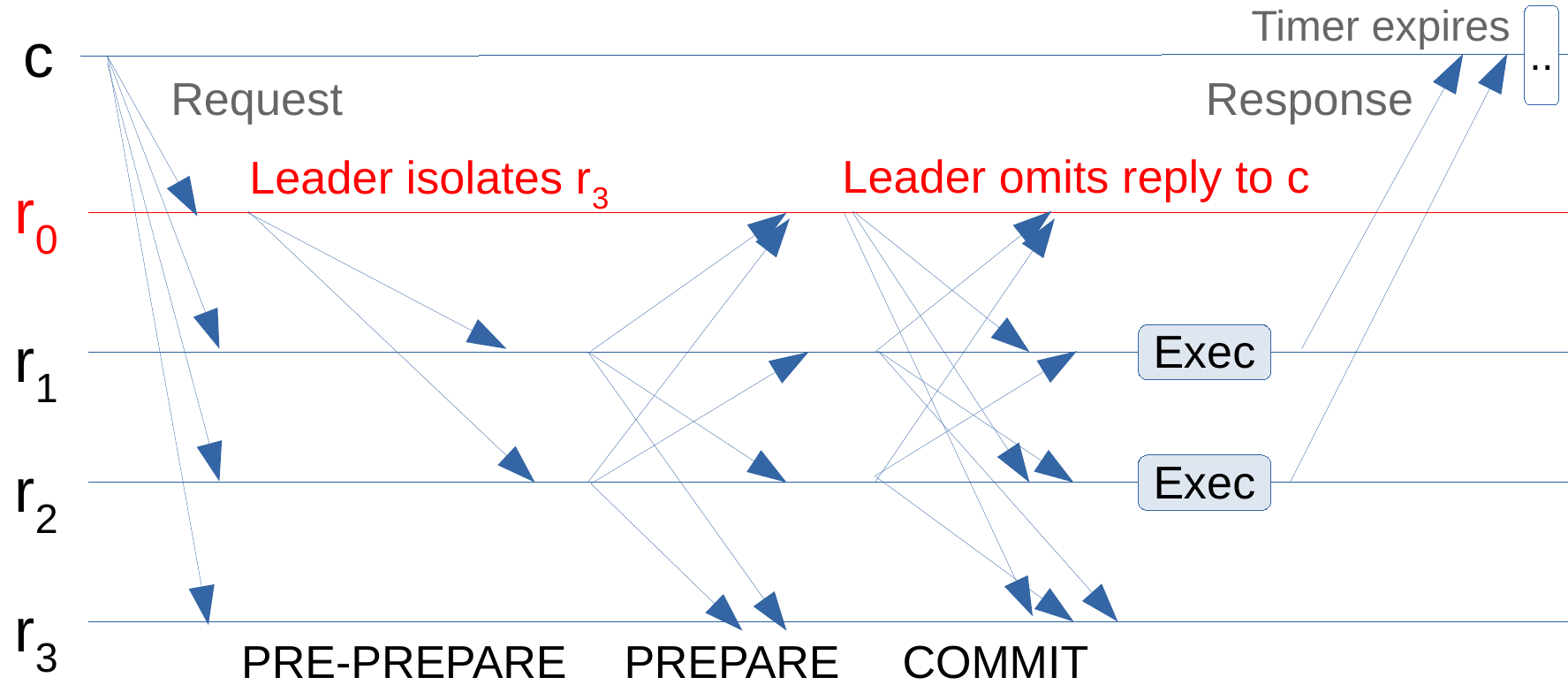}
    \caption{In the isolating leader attack, a Byzantine leader separates up to $f$ correct replicas from participation in the normal case operation pattern, thus preventing clients to accept a reply if a response quorum of $2f+1$ is used.}
    \label{fig:attack}
\end{figure}

\subsection{Breaking PBFT Liveness} \label{pbft-attack}

The attack (see Figure~\ref{fig:attack}) starts with a client broadcasting a normal update request.
The faulty leader, who wants to block this client from receiving enough replies (a form of censorship), purposely isolates up to $f$ correct replicas by selectively not sending a \texttt{PRE-PREPARE} message to them. 
For the other replicas, the leader acts correctly, following the normal case message flow. Since the \texttt{PREPARE} messages do not contain the proposed value $v$, but only a hash of it, the isolated correct replicas have no way of knowing the proposal as it was not reliably broadcasted to them before.
In particular, isolated replicas will not be able to prepare $v$.
This leads to the group of isolated replicas not being able to participate in the protocol anymore.
After the \texttt{COMMIT} phase, the leader executes the request batch together with $2f$ other replicas, being up to $f$ among them faulty (including the leader).
The attack proceeds with the leader not responding back to the client, making it gather between $f+1$ and $2f$ results (other Byzantine replicas within the group of $2f+1$ committing replicas might not respond as well). 

Due to the use of the read-only optimization, the client cannot accept a reply with less than $2f+1$ matching responses.
As a consequence, it will re-transmit the request after a timeout.
In this case, all correct replicas that executed the request (between $f+1$ and $2f$) re-transmit their result, while the Byzantine leader and other faulty replicas continue to selectively not respond to the client. Note that correct replicas that did not execute the request will not respond to the client.

When it comes to a faulty leader, many problems in PBFT are solved by performing a \textit{view change}, which replaces the leader and synchronizes replicas.
However, the group of isolated and correct $f$ replicas cannot trigger the view change protocol as long as the Byzantine leader displays his misbehavior only to these replicas. 
This happens because of the requirement of having $f+1$ \texttt{VIEW-CHANGE} messages for demoting a leader.
This requirement is intentional, as a group of $f$ Byzantine replicas should not be able to disturb the stability of the system by maliciously triggering view changes.

Even state transfer cannot help the affected client for a number of reasons.
First, a single client issuing an operation has no guarantee that there will be subsequent operations (and a sufficient number of such operations) by other clients so that replicas will perform a checkpoint and thus a state-transfer is  eventually being triggered. Second, the $f$ correct replicas might recover an updated state with the effect of the client operation but they might not be able to answer the client since the operation was not executed by them.\footnote{Note that if the PBFT implementation does not write client replies in the checkpoint \cite{distler2021byzantine}, or discards replies to save memory~\cite{castro2002practical}, then there is no guarantee the complete replica state necessary to respond to pending requests is actually being transferred, so these requests might not complete.} Third, state transfer happens rather infrequently (e.g., after a few minutes if not hours), and making a client wait this long just to complete a single request is impractical.


In the end, the $f+1$ non-isolated, correct replicas that participate in the protocol run cannot distinguish this attack from a successful normal case operation.
Re-transmissions do not solve the issue either as long as the Byzantine leader perpetuates the attack.


The described attack does not work in the basic, non-optimized variant of PBFT in which a client only waits for $f+1$ matching replies to accept a result.
The problem appears only when we increase the reply quorum size for incorporating fast read-only operations preserving linearizability.

In PBFT's liveness proof sketch (there's no formal proof of liveness for this protocol)~\cite{castro2001byzantine}, the assumption that a client only requires $f+1$ replies is employed to show operations eventually complete.
It is noteworthy that in PBFT it is only guaranteed that $2f+1$ replicas need to participate for committing a request, but it is not guaranteed that these are all \textit{correct} replicas responding to a client.
It means that PBFT does not ensure that all correct replicas execute all operations in the same order, only that eventually their state will reflect all these operations.

The main reason for this is the use of state transfer to deal with lossy channels and finite memory, which basically imposes the need for such a mechanism for recovering servers.


\subsection{The Attack Against BFT-SMaRt} \label{bft-smart-problem}

BFT-SMaRt implements Mod-SMaRt, a modular SMR protocol that can be instantiated with any consensus protocol as long as it satisfies some constraints~\cite{sousa2012byzantine}. 
The current implementation uses a version of a Byzantine Paxos~\cite{cachin2009yet} that makes the combined protocol similar to PBFT, at least in the normal case (see Figure~\ref{fig:bft}).

BFT-SMaRt also implements PBFT's read-only optimization, and it makes the system vulnerable to the attack just described for PBFT.
The Mod-SMaRt liveness proof is based on the following claim (last paragraph of Theorem 2's proof~\cite{sousa2012byzantine}):
``\emph{... the consensus instance will eventually decide ... in at least $n-f$ replicas. Because out of this set of replicas there must be $f+1$ correct ones, the operation will be correctly ordered and executed in such replicas.}''
Notice that, in this case, there is no consensus protocol that can solve the issue, since Mod-SMaRt itself only ensures $n-f$ replicas will start the consensus protocol.

\subsection{Summary} 

In summary, the use of PBFT's read-only optimization in BFT systems can create a vulnerability in which a malicious leader can violate the liveness of the protocol.
Given the normal case operation, the main challenge is to guarantee that a quorum of $q > 2n/3$ matching replies will arrive at the client for any executed operation, which in turn demands the execution of ordered requests in at least $q$ \textit{correct} replicas and that replicas are able to respond pending operations even after a state transfer. 
An abstract view of this problem is illustrated in Figure~\ref{fig:problem_abstract}.
    
\begin{figure}[tb]
    \centering
    \includegraphics[width=1\linewidth]{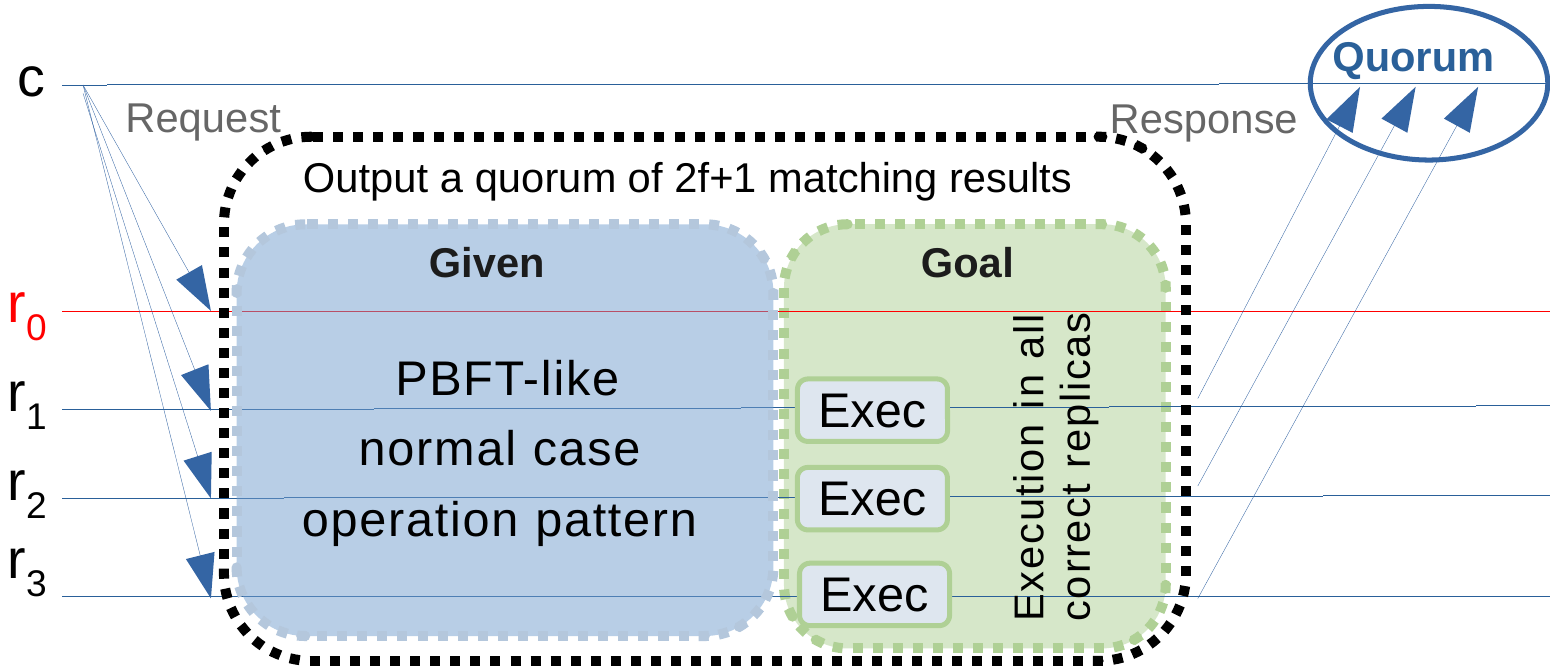}
    \caption{Abstract view of the problem.}
    \label{fig:problem_abstract}
\end{figure}

\sketch{ab: I would remove Figure~\ref{fig:problem_abstract} and the last phrase.}

\section{Patching the Protocols}
\label{solution}

To solve the problem, we propose a set of protocol modifications that prevent the isolation of correct replicas and ensure that if one correct replica executes a request, all correct replicas can execute, or at least send a reply for, such request.
This way, the client is guaranteed to receive a quorum of matching replies, which is necessary to accept a result when the read-only optimization is used.

In the following, we first provide a simple and straightforward solution to satisfy \textit{atomic broadcast' Agreement} property: the idea is to let every replica broadcast its decision to all others.
We subsequently present a second solution that improves the first approach by enabling replicas to query and forward decisions \textit{on demand}. 
This way we can avoid additional overhead during expected common fault-free executions.
Finally, we discuss how the checkpoints exchanged during state transfer need to be enriched to ensure our modification can be implemented with bounded memory.

For simplicity of exposition, the presented protocols assume \textit{reliable channels}, which can be implemented in our system model.
If these are not given, the protocols need to be modified to keep re-transmitting messages.

Moreover, for interested readers we additionally provide correctness proofs for the modifications we made in the appendix of this paper.

\subsection{Broadcasting Decisions} 
\label{broadcast-decisions}

To ensure a decision is made in all correct replicas, we can extend the 3-phase normal case operation message pattern with an additional all-to-all broadcast phase (see  Figure~\ref{fig:solution:decision-broadcasting}, and Algorithm~\ref{algo:Dec-Broadcasting}).

When deciding a batch of requests $v$ in consensus instance (or sequence number) $c$, each replica $r \in R$ broadcasts the decision along with a \textit{cryptographic proof} $\Gamma$ to all other replicas by sending a $\langle \texttt{FWD-DECISION}, c, v, \Gamma \rangle$ message.
This procedure has the following requirements:

    
   (1)~$\Gamma$ needs to be \textit{externally} verifiable, e.g., a set containing a quorum ($2f+1$ in a typical system with $n=3f+1$ replicas)
   of signed \texttt{COMMIT} messages in PBFT or 
   signed \texttt{ACCEPT} messages in BFT-SMaRt that validate the decision. Further, 
   the function \texttt{verify($c$,$v$,$\Gamma$)} checks if the decision is valid.
    
    (2)~If some replica $r'$ has not decided consensus $c$ and receives a $\langle \texttt{FWD-DECISION}, c, v, \Gamma \rangle$ message, it checks that the proof is valid and decides $v$, executing its requests when the requests decided in all previous instances are executed.


\begin{algorithm}[tb] 
\small

 \caption{Decision-Broadcasting at replica $r$}
\DontPrintSemicolon
 \label{algo:Dec-Broadcasting}
 
\SetKwProg{Upon}{Upon}{ do}{}
\SetKwProg{Function}{Function}{}{}
\SetKwComment{Comment}{$\triangleright$\ }{}
\Function{\textup{decide($c, v$)}}{
\If{\textup{$\neg$ isDecided($c$)}}
{

\texttt{decided$_c$} $\leftarrow$ \texttt{true}\;
retrieve $\langle c, v, \Gamma \rangle$ from \texttt{Log}\;
broadcast $\langle \texttt{FWD-DECISION}, c, v, \Gamma \rangle$ to all\;
deliver($c, v$) \Comment*[r]{Execution}
}
}

\Upon{\textup{reception of a $\langle \texttt{FWD-DECISION}, c, v, \Gamma \rangle$}}
{
\If{\textup{$\neg$ isDecided($c$) 	$\land$ verify($c,v,\Gamma$)}}
{
\texttt{Log} $\leftarrow$ \texttt{Log} $\cup \{ \langle c, v, \Gamma \rangle \} $ \;
decide($c,v$)\;
}
}
\end{algorithm}

\begin{algorithm}[tb] 
\small

 \caption{Decision-Forwarding at replica $r$}
\DontPrintSemicolon
 \label{algo:Forwarding}
 
\SetKwProg{Upon}{Upon}{ do}{}
\SetKwProg{Create}{Create}{:}{}
\SetKwProg{Function}{Function}{}{}
\SetKwComment{Comment}{$\triangleright$\ }{}
\Function{\textup{decide($c, v$)}}{
\If{\textup{$\neg$ isDecided($c$)}}
{

\texttt{decided$_c$} $\leftarrow$ \texttt{true}\;
deliver($c, v$) \Comment*[r]{Execution}
}
}

\Upon{\textup{reception of $f+1$ $\langle \texttt{ACCEPT}, c, h \rangle $ for hash $h$}}{

\If{\textup{Proposal($c$) = $\bot$ $\lor$ hash(Proposal($c$)) $\neq h$}}
{
send $\langle \texttt{REQ-DECISION}, c \rangle$ to $2f$ other replicas\;
}

}

 \Upon{\textup{reception of a $\langle \texttt{REQ-DECISION}, c \rangle$ from $r'$}}{\Create{\textup{new EventHandler}}{\Upon{\textup{isDecided($c$)}}{
 retrieve $\langle c, v, \Gamma \rangle$ from \texttt{Log}\;
 send $\langle \texttt{FWD-DECISION}, c, v, \Gamma \rangle$ to $r'$}}
}

\Upon{\textup{reception of a $\langle \texttt{FWD-DECISION}, c, v, \Gamma \rangle$}}
{
\If{\textup{$\neg$ isDecided($c$) 	$\land$ verify($c,v,\Gamma$)}}
{
\texttt{Log} $\leftarrow$ \texttt{Log} $\cup \{ \langle c, v, \Gamma \rangle \} $ \;
broadcast $\langle \texttt{FWD-DECISION}, c, v, \Gamma \rangle$ to all\;
decide($c$,$v$)\;
}
}
\end{algorithm}

These requirements enable a correct replica to decide by either gathering enough \texttt{COMMIT}/\texttt{ACCEPT} messages for a prepared leader proposal, or by processing a forwarded decision that comes with a valid proof.
The \textit{decision broadcast} phase ensures that, if a single correct replica reaches a decision, it propagates to all other correct replicas.
In fault-free execution, replicas may still commit after the usual three phases.

\sketch{ab: this solution is only ensured to work if we have reliable channels.}

\begin{figure*}[tb]
    \centering
        \begin{subfigure}[b]{.495\textwidth}
     \includegraphics[width=1\textwidth]{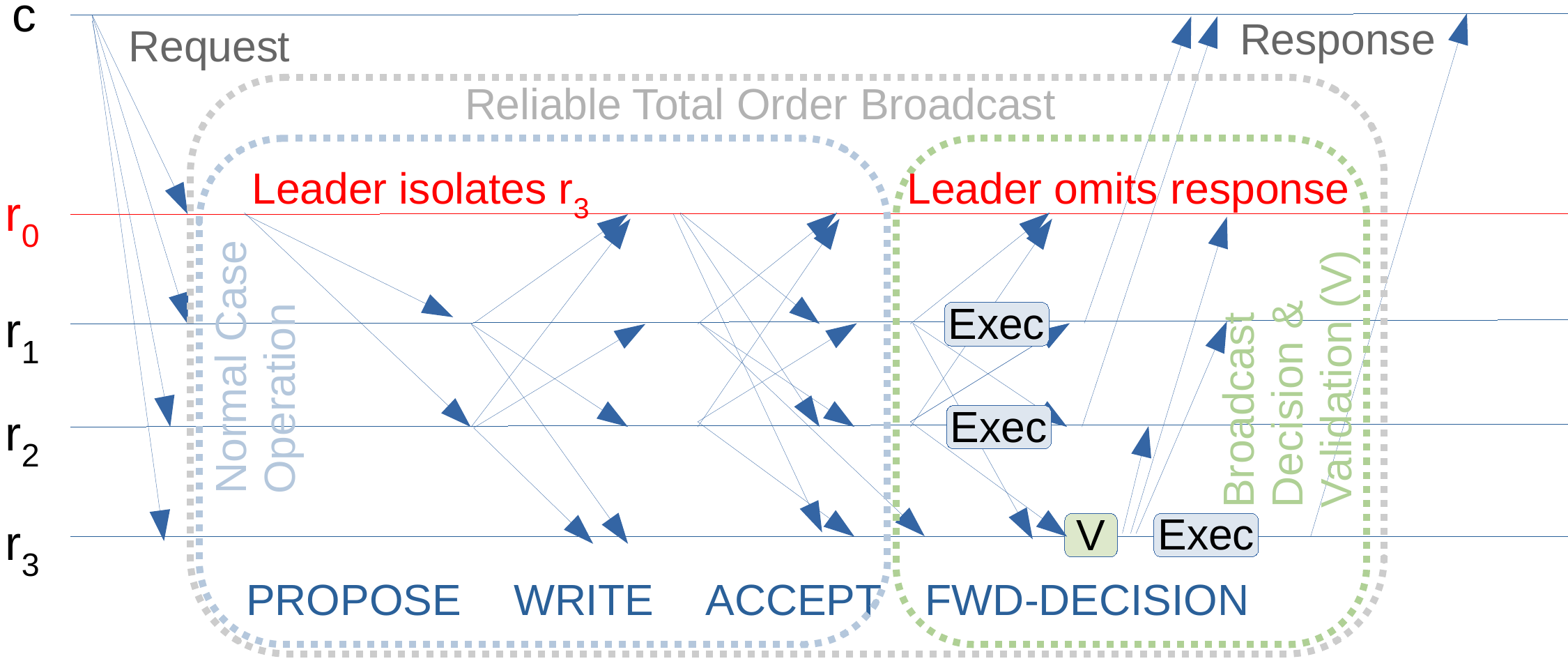}
    \caption{Decision Broadcasting: the normal case operation pattern is extended by an additional all-to-all broadcast of \texttt{FWD-DECISION}.}
       \label{fig:solution:decision-broadcasting}
    \end{subfigure}
    \begin{subfigure}[b]{.495\textwidth}
            \includegraphics[width=1\textwidth]{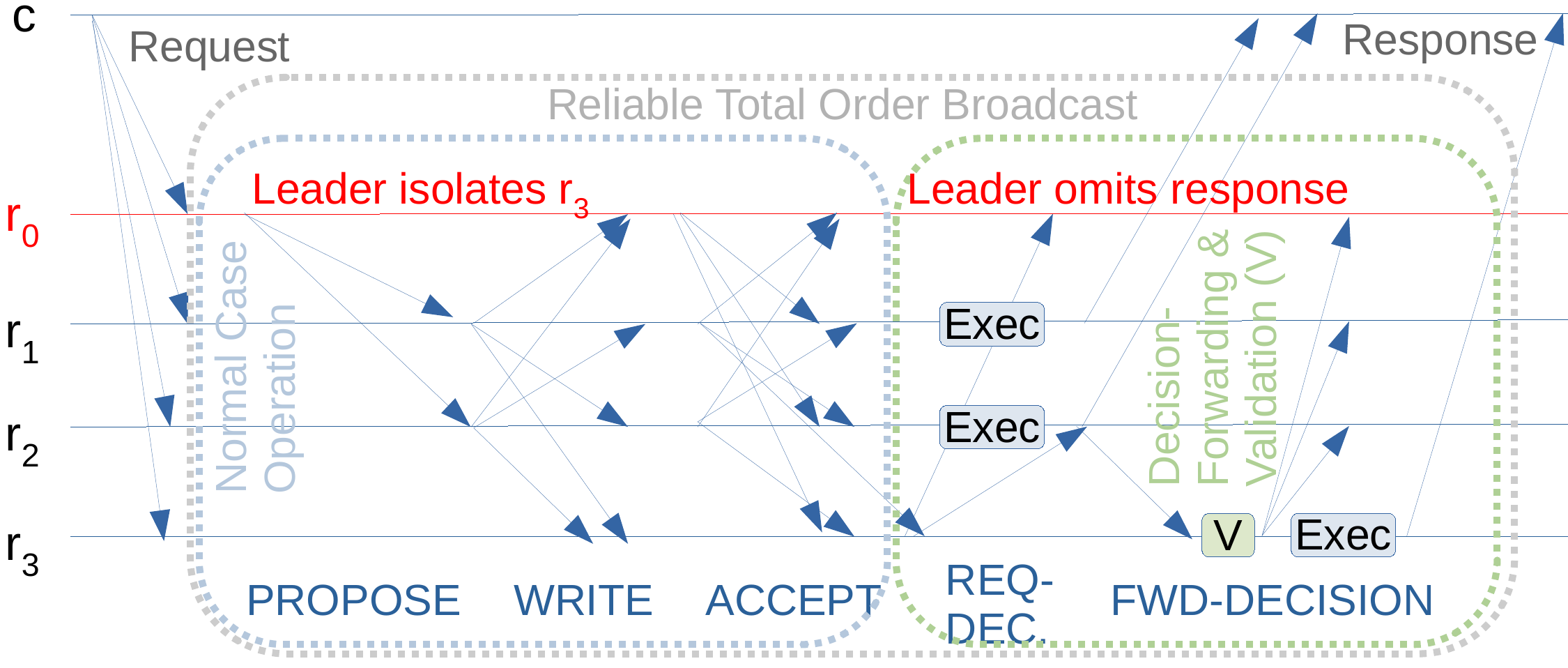}
    \caption{Decision Forwarding: the normal case operation pattern is extended \textit{on demand} by \texttt{REQ-DECISION} and \texttt{FWD-DECISION}. 
 }
       \label{fig:solution:decision-forwarding}
    \end{subfigure}
    \caption{
    Correct replicas can propagate a decision and externally-verifyable proof $\Gamma$ to possibly isolated replicas.
    }
    \label{fig:solution}
\end{figure*}

\subsection{Forwarding Decisions on Demand}
\label{decision-forwarding}

While the \textit{decision broadcasting} solution works, it puts an additional burden on the system during \textit{fault-free} executions.
Replicas disseminate the decision as soon as possible to make it available to the others -- even when this is unnecessary.
To avoid the induced overhead of broadcasting decisions when the leader is correct, we aim to make this scheme \textit{on demand}.
Correct replicas can be queried by other replicas that are isolated, to forward a decision to them.
Therefore, we extend the 3-phase consensus protocol with two additional messages: \texttt{REQ-DECISION} and \texttt{FWD-DECISION}.
The first one is used by (potentially stuck) replicas to query a decision along with a cryptographic proof $\Gamma$ from other replicas, while the second one is the corresponding reply message that contains the forwarded decision and proof.
In the following, we briefly describe this protocol extension for BFT-SMaRt (see Algorithm~\ref{algo:Forwarding} - the overall pattern is illustrated in Figure~\ref{fig:solution:decision-forwarding}):

A correct replica $r$ requests a decision from others for consensus $c$ as soon as it receives $f+1$ $\langle \texttt{ACCEPT}, c, h \rangle$ for some hash $h$ \textit{and} it has not yet received a leader proposal or the received proposal does not match $h$ (Lines 5--7). Note that a malicious leader can also isolate $r$ by sending a corrupt or deviating proposal. This requirement ensures that such a request is sent if (1)~a decision is made (this causes $r$ to receive  at least $f+1$ \texttt{ACCEPT}s from correct replicas) and (2)~$r$ was isolated by the leader. Note that it is also possible that a request is sent under a correct leader if the proposal arrives late at $r$. This is not a problem because asking for a decision does not block $r$'s participation in the normal protocol.

If some correct replica $r'$ receives such a request, it will respond back to $r$ as soon as a decision becomes available (Lines 8-12).
Further, it is necessary that $r$ asks sufficiently many other replicas to reach at least one of the $f+1$ correct replicas that are guaranteed to eventually decide.
Moreover, in our solution replicas also remember to whom they forwarded decisions per consensus (not shown in the simplified pseudo code), so that decisions can only be queried once per replica and consensus.

Finally, if a forwarded decision is received by $r$, then $r$ first checks if it is still undecided in consensus $c$ (to preserve \textit{Integrity} of consensus) and subsequently uses the proof $\Gamma$ to verify the correctness (\textit{Agreement}) of the received consensus decision (Lines 13--14).
The decision is appended to the log (Line 15), and echoed to all other replicas (Line 16).
This ``echo'' ensures the \textit{Termination} property of consensus holds even if some correct replica uses a forwarded decision to decide instead of further participating in the three-step protocol.
This happens because if a single correct replica $r$ decides using a forwarded decision message $d$, all correct replicas will also receive $d$ and eventually decide.


\subsection{Implications on State Transfer}

The protocol described in the previous section assumes that decisions and their proofs are stored forever in case some replica needs them to reply to a \texttt{REQ-DECISION} message. 
Both PBFT and BFT-SMaRt include logging features that can be extended to also log this information.
%
%
To be practical, however, such logs need to be prevented from growing in an unbounded way.
PBFT and BFT-SMaRt achieve this by using a garbage collection mechanism to create state checkpoints and delete all log entries belonging to preceding decisions.

In case a replica receives a forwarding request for a decision it has already garbage-collected, it will respond with an error message (\texttt{OUTDATED-REQ}). In such a case, the requesting replica has fallen back too far and thus needs to initiate a state transfer to catch up with the others. The error message should include a recent proof for the currently highest decided consensus to convince the requester of this fact.

An important detail, necessary for the correctness of this approach, is the following:
When using the read-only optimization and our solution we need to enrich the replica state with the last reply sent to each client.\footnote{Assuming clients operate in closed loop, with at most one pending request at a time. If a client can have up to $k$ pending requests, the state must store the last $k$ replies sent to the client.} 
These replies need to be transferred during state transfer to ensure the recovered replica can reply to pending requests issued by any correct client. 
If a reply needs to be re-transmitted because of a client timeout, then, assuming checkpoints would not contain a reply store, the replicas that fetched the state from the checkpoint cannot provide responses for the requests they skipped ~\cite{distler2021byzantine}. 
If the last replies to clients are transferred during state transfer, there will be eventually $2f+1$ correct replicas that can respond to clients.

\section{Implementation and Evaluation}
\label{evaluation}

In this section, we first give a brief overview of our LAN and WAN experimental setups and the implemented solutions.
Subsequently, we quantify to which extent the read-only optimization is worth having in BFT systems by comparing the performance of optimized and non-optimized read-only operations.
After that, we compare the performance of our implemented solutions against the normal BFT-SMaRt as baseline to reason about the induced overhead.
Lastly, we compare \textit{decision-broadcasting} (Section~\ref{broadcast-decisions}) and \textit{decision-forwarding} (Section~\ref{decision-forwarding}) to evaluate their performance during fault-free operation and when an isolation attack occurs. 

\subsection{Implemented Solutions and Experimental Setup} 
\label{implementation}

We implemented the attack described in Section~\ref{pbft-attack} and both \textit{decision broadcasting} and \textit{decision forwarding} as solutions in BFT-SMaRt.
We define the variants \textit{decision broadcasting*} and \textit{decision forwarding*} as the respective counterparts of our implemented solutions in which the leader perpetually conducts the isolation attack.
We use them to investigate the impact of the attack on the system performance.

\subsubsection*{LAN Setup}
To benchmark our prototype in terms of throughput and latency, we use a setup in which experiments are conducted with 4 replicas hosted on separate Ubuntu 20.04 VMs with 4 vCPUs and 10 GB of RAM, each on an Intel Xeon Silver 4114-based server. Clients are distributed on 4 separate VMs with the same specification. We use the micro-benchmarks of BFT-SMaRt to measure throughput and latency. Throughput is measured at the leader replica, while latency is measured on the client side at a chosen client. Further, we conduct measurements with an increasing number of clients (up to 400) to generate higher workloads, thus cautiously approaching the limits of a deployed system until we reach saturation.

\subsubsection*{WAN Setup}
To benchmark the latency of a system when deployed in a wide-area network, we use the Amazon AWS cloud for placing several EC2 instances in different regions.
As we have no high hardware requirements for our latency experiments, we employ a \textit{t2.micro} instance type.
This type has 1 vCPU, 1 GB of RAM, and 8 GB standard SSD volume (gp2).
We choose the regions \textit{Virginia}, \textit{Ireland}, \textit{São Paulo} and \textit{Sydney} to place a dedicated client VM and a dedicated server VM in each region, with the leader being placed in Virgina.
Clients send requests simultaneously and we run sufficiently enough clients to reach a throughput of at least 100 requests per second. 
Further, clients measure latency as the time between sending a request and delivering a response for it.
Each client samples $1000$ requests with a payload size of 100 bytes. There is one measuring client in each region.

\subsection{Is the Read-only Optimization Worth It?}
\label{read-only-evaluation}

An important question to answer is the following one: \textit{is incorporating the read-only optimization into a BFT system really worth the trouble?} 
The benefits we may expect are twofold.
First, since read-only requests can skip the three-step agreement protocol, there are fewer requests that need to be ordered, thus requiring less coordination between replicas, which may lead to higher performance.

Second, in wide-area deployments, latency between replicas can be high, and thus coordination is expensive in terms of latency. Employing a two-phase request-response pattern rather than needing a total of 5 communication steps (3 additional communication steps are needed for ordering) may achieve substantially faster response times for read operations.

\subsubsection*{LAN Evaluation}
We first measured the performance of (optimized and non-optimized) read and write operations in our local setup for a payload of 4kB, with results shown in  Figure~\ref{fig:reads_vs_writes}. 
Here, we observe that with optimized reads, we can consistently achieve lower latencies at higher throughput levels, indicating that performance-wise a BFT system can highly benefit from supporting the read-only optimization. 

 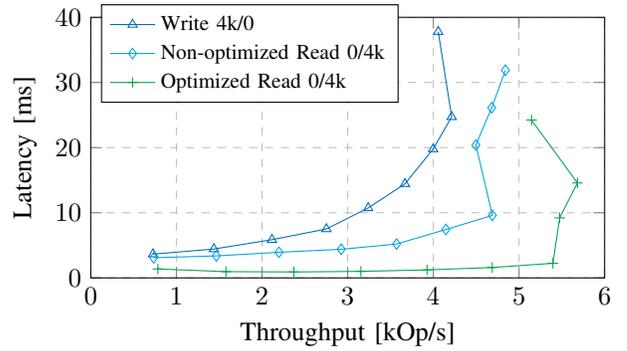
\begin{figure}[tb]
    \centering
     \begin{tikzpicture}
    \begin{axis}[
width=0.95\columnwidth,
height=0.57\columnwidth,
    xlabel={Throughput [kOp/s]},
    ylabel={Latency [ms]},
    xmin=0, xmax=6,
    ymin=0, ymax=40,
    xtick={0,1,2,3,4,5,6,7,8,9,10},
    ytick={0,10, 20 ,30, 40},
    legend pos=north west,
    legend style={at={(0.02 ,1.05)}},
    legend cell align={left},
    ymajorgrids=true,
    xmajorgrids=true,
    grid style=dashed,
    legend style={font=\footnotesize},
]

\addplot[
    color=NavyBlue,
    mark=triangle,
    ]
    table [x=throughput,y=latency] {data/writes_4k_0.dat};

    \addplot[
    color=Cerulean,
    mark=diamond,
    ]
    table [x=throughput,y=latency] {data/writes_0_4k.dat};
    
        \addplot[
    color=Green,
    mark=+,
    ]
    table [x=throughput,y=latency] {data/reads_0_4k.dat};

   
   \legend{Write 4k/0, Non-optimized Read 0/4k , Optimized Read 0/4k}
\end{axis}
\end{tikzpicture} 
    \caption{Performance comparison between read and write operations in BFT-SMaRt for request / response sizes (in byte).} 
    \label{fig:reads_vs_writes}
\end{figure}

 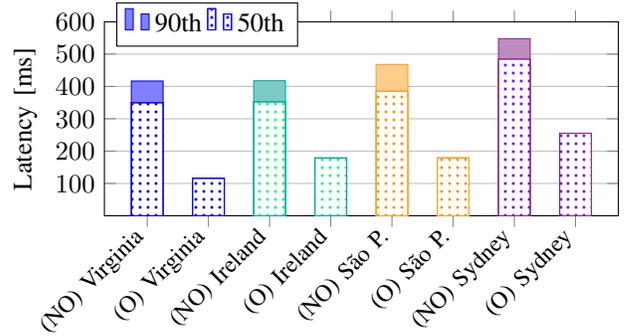
\begin{figure}[tb]
    \centering
    \begin{tikzpicture}
	\begin{axis}[ybar,
	bar shift=0pt,
    ymin=0,
	width=0.95\columnwidth,
    height=0.47\columnwidth,
legend style={at={(0.2,1.1)},
anchor=north,legend columns=-1},
    ylabel={Latency [ms]},
       ymajorgrids=true,
    symbolic x coords={ (NO) Virginia, (O) Virginia, (NO) Ireland, (O) Ireland, (NO) São P., (O) São P., (NO) Sydney, (O) Sydney},
    bar width=12pt,
    ytick = {100,200,300,400,500,600},
    xtick = {(NO) Virginia, (O) Virginia, (NO) Ireland, (O) Ireland, (NO) São P., (O) São P., (NO) Sydney, (O) Sydney},
    ymin= 0, ymax=600,
    x tick label style={rotate=45,anchor=east},
    xticklabel style = {font=\small},
	]

   	\addplot[draw=blue,fill=blue!50] coordinates
		{((O) Virginia,116.0922216) 
		 ((NO) Virginia, 416.7311044) };
		 
    \addplot[draw=blue,fill=white, postaction={
        pattern = dots,
        pattern color=blue
    }] coordinates
		{((O) Virginia,115.9397975) 
		 ((NO) Virginia, 349.4849055) };

			\addplot[draw=JungleGreen,fill=JungleGreen!50] coordinates 
	 	{((O) Ireland, 178.8320836) 
	 	 ((NO) Ireland, 417.6423582) };
	 	
 \addplot[draw=JungleGreen,fill=white, postaction={
        pattern = dots,
        pattern color=JungleGreen
    }] coordinates
		{((O) Ireland, 178.6414355) 
		 ((NO) Ireland, 352.315769) };
		
		   	\addplot[draw=YellowOrange,fill=YellowOrange!50] coordinates
		{((O) São P., 179.0964192) ((NO) São P., 468.0573602) };
    \addplot[draw=YellowOrange,fill=white, postaction={
        pattern = dots,
        pattern color=YellowOrange
    }] coordinates
		{((O) São P., 178.9278565) ((NO) São P., 385.252602) };

  	\addplot[draw=Plum,fill=Plum!50] coordinates
		{((O) Sydney, 255.2791092) ((NO) Sydney, 547.7926776) };
    \addplot[draw=Plum,fill=white, postaction={
        pattern = dots,
        pattern color=Plum
    }] coordinates
		{((O) Sydney, 255.1146705) ((NO) Sydney, 484.542422) };
	\legend{90th, 50th}
	\end{axis}
\end{tikzpicture} 
     \vspace{-1.08em}
    \caption{Latency comparison between optimized (O) and non-optimized (NO) read operations in BFT-SMaRt deployed in a WAN, as observed by clients in different regions.}
    \label{fig:reads_vs_writes_wan}
\end{figure}

\sketch{ab: it would be good to remove Read 4K/0 (it doesn't represent anything) and I would change 'Write 0/4K' to 'Non-optimized Read 0/4K'. Then we can cut the phrase I added from the text.}

\subsubsection*{Wide-area Evaluation}

We employed our WAN setup to investigate how much faster an optimized read is when compared with its non-optimized counterpart, which follows the three-step consensus protocol.
For the sake of a fair comparison, we are interested in the ``opportunity costs'', that is, latency gains that can be achieved when implementing the optimization. 
Thus, for the non-optimized read operations, we use the standard variant of BFT-SMaRt that orders all operations and allows clients to accept results as soon as $f+1$ matching replies are obtained, while in the optimized variant, a quorum of matching replies is needed.
This makes a slight difference: 
in a WAN setup, waiting for the fastest $f+1$ rather than the fastest $q > 2n/3$ out of $n$ replicas brings some advantage in terms of latency. 
Figure~\ref{fig:reads_vs_writes_wan} shows the median and 90th percentile of measured latencies of a $1000$ requests sample for a client in each region.

The results show that optimized reads are substantially faster than non-optimized operations. 
In particular, averaging across all sites, non-optimized reads have a latency $2.16\times$ higher than the optimized variant (median), and this difference increases to $2.54\times$ when considering the 90th percentile.
Furthermore, non-optimized reads display a higher variation in observed latency, with the 90th percentile of latencies being considerably higher than the median, when compared with optimized reads.
This can be explained by a random waiting time of a write request that arrives at the leader.
The leader needs to wait for an ongoing consensus instance to finish before it can put the pending request in a batch and propose it.
The expected waiting time is half of the consensus latency on BFT-SMaRt~\cite{berger20aware}.

\begin{tcolorbox}[colback=light-gray, boxrule=0pt, top=6pt, left=6pt, right=6pt, bottom=6pt]

Considering the obvious advantages of the read-only optimization for both local and WAN deployments, we conclude that it is worthwhile to have this optimization implemented in practical BFT SMR systems.

\end{tcolorbox}

\subsection{Comparison of Implemented Solutions}
\label{comparision}

Both \textit{decision broadcasting} and \textit{decision forwarding} modify the ordering phase. We want to investigate the impact their implementation has on performance. For this reason, we compare the original BFT-SMaRt system as a baseline with our implemented solutions. 

\subsubsection*{Fault-free Execution}
First, we evaluate the performance of the fault-free execution. We use our local setup and a request / response size of 4~kB/4~kB. We show the results of our measurements in  Figure~\ref{fig:performance_comparison}.

\sketch{ab: 4kB/4kB is supposed to represent what? 4kB/0 would be better}

 \begin{figure}[tb]
    \centering
     \begin{tikzpicture}
    \begin{axis}[
width=0.95\columnwidth,
height=0.57\columnwidth,
    xlabel={Throughput [kOp/s]},
    ylabel={Latency [ms]},
    xmin=0, xmax=4,
    ymin=0, ymax=40,
    xtick={0,1,2,3,4,5,6,7,8,9,10},
    ytick={0,10, 20 , 30, 40},
    legend pos=north west,
        legend cell align={left},
    ymajorgrids=true,
    xmajorgrids=true,
    grid style=dashed,
    legend style={font=\footnotesize},
]

\addplot[
    color=blue,
    mark=triangle,
    ]
    table [x=throughput,y=latency] {data/dec_broadcasting.dat};

        \addplot[
    color=red,
    mark=+,
    ]
    table [x=throughput,y=latency] {data/dec_fwd.dat};

        \addplot[
    color=PineGreen,
    mark=square,
    ]
    table [x=throughput,y=latency] {data/bftsmart.dat};
   
   \legend{decision broadcasting, decision forwarding,  standard BFT-SMaRt}
\end{axis}
\end{tikzpicture} 
     \vspace{-1.08em}
    \caption{Performance comparison in fault-free operation.}
    \label{fig:performance_comparison}
\end{figure}
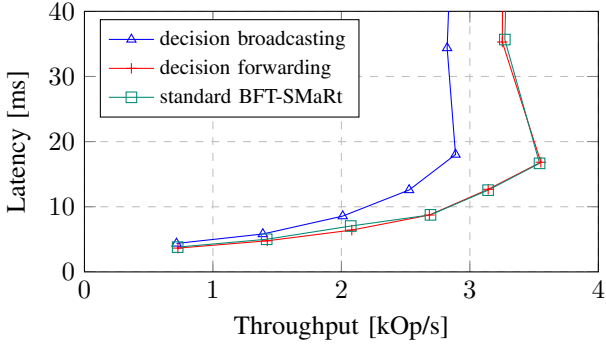

It turns out that decision broadcasting decreases observed performance significantly, while decision forwarding seems to display roughly the same performance as the original BFT-SMaRt.
Both observations are expected: always broadcasting decisions as an additional phase induces overhead in the system. 
The message pattern of decision forwarding corresponds to the normal pattern if no decision is requested. This is almost always the case under a correct leader, so the decision forwarding variant can achieve similar performance. 


\subsubsection*{Performance when under Attack}
In the next experiment, we investigate which effect a Byzantine leader that conducts an isolation attack has on the system.
First, we compare the performance of decision broadcasting* and decision forwarding* in our local setup (*-variants means \textit{under attack}).
We show the obtained results in Figure~\ref{fig:performance_comparison_attacked}.
Notice that we cannot evaluate the performance of the standard BFT-SMaRt when under attack, since it would prevent clients from accepting the operation result.

As we can see, the performance of decision forwarding* is now much closer to the performance of decision broadcasting*, with the former being only slightly better.
This is because when forwarding is used, two replicas are always probed to forward a decision, while in broadcasting all replicas disseminate the decision.
On an interesting side note, the performance actually increases when compared to the fault-free experiment.
The reason for this is that the isolation attack causes the leader to not propose to replica 3 and to not respond to the (up to 400) clients, thus relieving a lot of effort from the leader. Since the leader is typically the bottleneck, doing less work leads to increased performance. This applies to both decision broadcasting and forwarding when under attack. But for decision forwarding* this effect becomes less visible in the diagram since replicas now need to do extra work by forwarding decisions, which impacts performance negatively. 

\begin{figure}[!tb]
    \centering
      \begin{tikzpicture}
    \begin{axis}[
width=0.95\columnwidth,
height=0.57\columnwidth,
    xlabel={Throughput [kOp/s]},
    ylabel={Latency [ms]},
    xmin=0, xmax=4,
    ymin=0, ymax=40,
    xtick={0,1,2,3,4,5,6,7,8,9,10},
    ytick={0,10, 20 , 30, 40},
    legend pos=north west,
        legend cell align={left},
    ymajorgrids=true,
    xmajorgrids=true,
    grid style=dashed,
    legend style={font=\footnotesize},
]

\addplot[
    color=blue,
    mark=triangle,
    ]
    table [x=throughput,y=latency] {data/dec_bcast_attacked.dat};

        \addplot[
    color=red,
    mark=+,
    ]
    table [x=throughput,y=latency] {data/dec_fwd_attacked.dat};

   \legend{decision broadcasting (attacked), decision forwarding (attacked)}
\end{axis}
\end{tikzpicture} 
     \vspace{-1.08em}
    \caption{Performance comparison when under attack.}
    \label{fig:performance_comparison_attacked}
\end{figure}
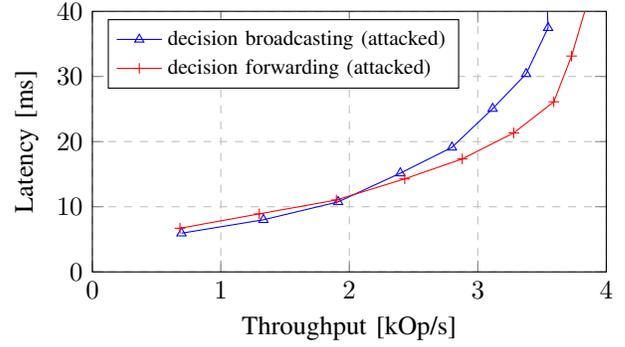

\begin{tcolorbox}[colback=light-gray, boxrule=0pt, top=6pt, left=6pt, right=6pt, bottom=6pt]

Decision forwarding performs better than decision broadcasting when  replicas are placed in a LAN. It avoids unnecessary overhead in a fault-free operation, and also achieves better performance during an isolation attack.

\end{tcolorbox}

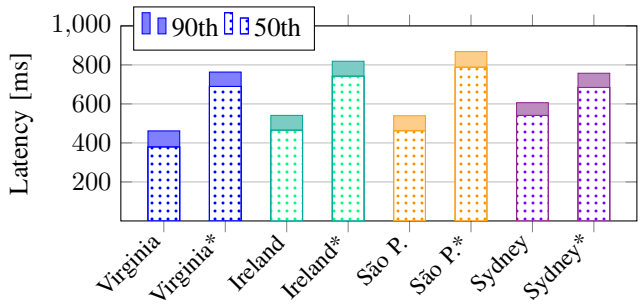
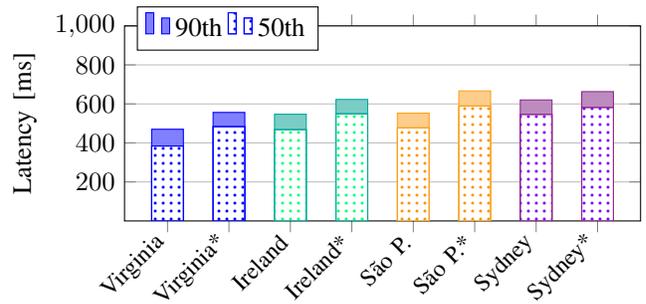
\begin{figure*}[tb]
    \centering
    \begin{subfigure}[b]{.49\textwidth}
		\begin{tikzpicture}
	\begin{axis}[ybar,
	bar shift=0pt,
    ymin=0,
	width=0.95\columnwidth,
    height=0.47\columnwidth,
legend style={at={(0.2,1.1)},
anchor=north,legend columns=-1},
    ylabel={Latency [ms]},
       ymajorgrids=true,
    symbolic x coords={Virginia, Virginia*, Ireland, Ireland*, São P., São P.*, Sydney, Sydney*},
    bar width=12pt,
    ytick = {200, 400, 600, 800, 1000},
    xtick = {Virginia, Virginia*, Ireland, Ireland*, São P., São P.*, Sydney, Sydney*},
    ymin= 0, ymax=1000,
    x tick label style={rotate=45,anchor=east},
    xticklabel style = {font=\small},
	]

   	\addplot[draw=blue,fill=blue!50] coordinates
		{(Virginia, 461.613517) 
		 (Virginia*, 763.292463) };
		 
    \addplot[draw=blue,fill=white, postaction={
        pattern = dots,
        pattern color=blue
    }] coordinates
		{(Virginia, 379.379647) 
		 (Virginia*, 689.422365) };

			\addplot[draw=JungleGreen,fill=JungleGreen!50] coordinates 
	 	{(Ireland, 541.0428058) 
	 	 (Ireland*, 818.8077777) };
	 	
 \addplot[draw=JungleGreen,fill=white, postaction={
        pattern = dots,
        pattern color=JungleGreen
    }] coordinates
		{(Ireland, 465.3432775) 
		 (Ireland*, 742.2045105) };

		   	\addplot[draw=YellowOrange,fill=YellowOrange!50] coordinates
		{(São P., 539.3239512) (São P.*, 868.3656947) };
		
    \addplot[draw=YellowOrange,fill=white, postaction={
        pattern = dots,
        pattern color=YellowOrange
    }] coordinates
		{(São P., 462.2330435) (São P.*, 787.9565115) };

  	\addplot[draw=Plum,fill=Plum!50] coordinates
		{(Sydney, 606.2821636) (Sydney*, 757.1850074) };
    \addplot[draw=Plum,fill=white, postaction={
        pattern = dots,
        pattern color=Plum
    }] coordinates
		{(Sydney, 540.7601335) (Sydney*, 684.2088015) };

	\legend{90th, 50th}
	\end{axis}
\end{tikzpicture} 
		 \vspace{-1.08em}
		\caption{\textit{decision forwarding} in normal operation and when attacked(*).}
		\label{fig:dec_fwd_wan} 
	\end{subfigure}
	\begin{subfigure}[b]{.49\textwidth}
		\begin{tikzpicture}
	\begin{axis}[ybar,
	bar shift=0pt,
    ymin=0,
	width=0.95\columnwidth,
    height=0.47\columnwidth,
legend style={at={(0.2,1.1)},
anchor=north,legend columns=-1},
    ylabel={Latency [ms]},
       ymajorgrids=true,
    symbolic x coords={Virginia, Virginia*, Ireland, Ireland*, São P., São P.*, Sydney, Sydney*},
    bar width=12pt,
    ytick = {200, 400, 600, 800, 1000},
    xtick = {Virginia, Virginia*, Ireland, Ireland*, São P., São P.*, Sydney, Sydney*},
    ymin= 0, ymax=1000,
    x tick label style={rotate=45,anchor=east},
    xticklabel style = {font=\small},
	]

   	\addplot[draw=blue,fill=blue!50] coordinates
		{(Virginia, 470.103887) 
		 (Virginia*, 556.3205958) };
		 
    \addplot[draw=blue,fill=white, postaction={
        pattern = dots,
        pattern color=blue
    }] coordinates
		{(Virginia, 384.074976) 
		 (Virginia*, 483.5200085) };

			\addplot[draw=JungleGreen,fill=JungleGreen!50] coordinates 
	 	{(Ireland, 547.2314004) 
	 	 (Ireland*, 622.8652112) };
	 	
 \addplot[draw=JungleGreen,fill=white, postaction={
        pattern = dots,
        pattern color=JungleGreen
    }] coordinates
		{(Ireland, 468.0768775) 
		 (Ireland*, 549.297033) };

		   	\addplot[draw=YellowOrange,fill=YellowOrange!50] coordinates
		{(São P., 552.8190852) (São P.*, 666.5969562) };
		
    \addplot[draw=YellowOrange,fill=white, postaction={
        pattern = dots,
        pattern color=YellowOrange
    }] coordinates
		{(São P., 477.2205685) (São P.*, 589.197462) };

  	\addplot[draw=Plum,fill=Plum!50] coordinates
		{(Sydney, 620.3523147) (Sydney*, 662.9971418) };
    \addplot[draw=Plum,fill=white, postaction={
        pattern = dots,
        pattern color=Plum
    }] coordinates
		{(Sydney, 546.007206) (Sydney*, 581.75605) };

	\legend{90th, 50th}
	\end{axis}
\end{tikzpicture}
	    \vspace{-1.08em}
		\caption{\textit{decision broadcasting} in normal operation and when attacked(*).}
		\label{fig:dec_bcast_wan} 
	\end{subfigure}
    \caption{Latency comparison of implemented variants in a WAN deployment as observed by clients located in different regions.}
    \label{fig:dec_variants_wan}
\end{figure*}

\subsubsection*{Wide-area Evaluation}
Finally, we want to see the implications our solutions have on BFT SMR systems that are deployed in WAN environments.
For this reason, we evaluate decision broadcasting and decision forwarding both in fault-free executions and when under attack (*-variants).
We use our WAN setup and show measurement results in Figure~\ref{fig:dec_variants_wan}.
 
We can observe that there is quite a perceivable difference in how latencies are affected when an isolation attack takes place.
When attacked, decision broadcasting performs better than decision forwarding across all sites of the system.
In particular, the median latency in decision forwarding increases by 57\% while in decision broadcasting it increases only by 18\%, when averaging across all sites.

 The explanation for this is that in decision broadcasting, replicas disseminate a decision as soon as possible to all others, while in decision forwarding a \texttt{REQ-DECISION} message is needed first, which causes an additional communication step that adds on the overall latency (see Figure~\ref{fig:dec_fwd_wan}).
 There is also a second reason: replicas do not decide at the same time but at different times. Since all replicas broadcast a decision, it is the fastest replica among them that determines the speed for the isolated replicas to catch up. This is quite a beneficial side-effect for decision broadcasting as can be seen in Figure~\ref{fig:dec_bcast_wan}.

\begin{tcolorbox}[colback=light-gray, boxrule=0pt, top=6pt, left=6pt, right=6pt, bottom=6pt]

Decision broadcasting appears to be a reasonable alternative for use in WANs.
When attacked, it avoids an extra communication step. The extra redundancy of letting all replicas broadcast the decision is beneficial (especially for small messages): the fastest broadcast decision arriving determines the tempo for isolated replicas to catch up.

\end{tcolorbox}
 
\section{Related Work}
\label{related-work}

We discuss several works on SMR read optimizations and further works that improve the robustness of BFT systems.


\subsubsection*{Optimizing Reads in CFT}
 
 
Chandra et al. introduced the idea of a \textit{leader lease}~\cite{chandra2007paxos}: as long as a stable Paxos leader has the lease, it is guaranteed to have up-to-date information in its local data, thus reads  can be served by the leader \textit{locally} and are never stale (satisfying linearizability). This idea is employed in several practical systems, e.g.,
Spanner~\cite{corbett2013spanner}, where \textit{snapshot reads} allow clients to read locally at specified timestamps, or Chubby~\cite{burrows2006chubby}.
 
The idea of leases allows a wide design space, e.g., besides leader leases it is also possible to grant leases to a set of replicas~\cite{moraru2014paxos} or even all~\cite{baker2011megastore}.
For instance, Moraru et al. presented \textit{Paxos Quorum Leases}~\cite{moraru2014paxos}, in which leases for different objects can be managed by individually chosen quorums of replicas.
Choosing a Paxos quorum for a lease helps with managing leases and reducing overhead as lease revocation and Paxos messages can be combined.

ZooKeeper~\cite{hunt2010zookeeper} allows fast local reads from \textit{any} server, but relaxes the consistency guarantee: it offers only sequential consistency but not linearizability, so stale reads can occur in the system.
Note that strict SMR requires linearizability.
  
\subsubsection*{Optimizing Reads in BFT}

The read-only optimization for BFT SMR that enables fast reads without requiring ordering and not sacrificing linearizability was first introduced by Castro and Liskov in PBFT~\cite{castro1999practical}. 
Several other BFT SMR systems that later followed also provide a read-only optimization, e.g.,
Q/U~\cite{abd2005fault},
HQ~\cite{cowling2006hq},
PBFT-CS~\cite{wester2009tolerating},
UpRight~\cite{clement2009upright},
Zyzzyva~\cite{kotla2007zyzzyva},
and BFT-SMaRt~\cite{bessani2014state} (and its variants~\cite{sousa2015separating, berger20aware}).

Further, to the best of our knowledge, the vulnerability we discussed in Section~\ref{bft-smart-problem} does not apply to all of these systems.
Instead, it exists in SMR designs that (1)~are agreement-based and employ the same normal-case ordering pattern as PBFT, which does not guarantee 
a sufficient number of replicas 
(needed for linearizability) 
can reply to the client, or (2)~do not enrich replica state by the last response sent to each client, so that after state transfer correct replicas are unable to respond to re-transmitted requests.
On a closer look, we think this vulnerability currently exists in several works (e.g., \cite{castro1999practical, cowling2006hq, wester2009tolerating, bessani2014state, sousa2015separating, berger20aware}).

It is noteworthy that there are BFT SMR designs with a correct read-only optimization.
An example is Zyzzyva~\cite{kotla2007zyzzyva}, which comes with a more sophisticated design than PBFT. It uses a combination of a forwarding technique (proposals are signed and can be queried from other replicas) and a sub-protocol for blaming a faulty leader:
if a Byzantine leader isolates a correct replica by sending conflicting proposals, then the correct replica can craft a proof of misbehavior (POM), which contains signed messages from the leader that reveal his misbehavior, and broadcasts it. Upon receiving such a POM, replicas trigger a view change and elect a new leader.
Note, that an isolation attack generally cannot be detected as long as the malicious leader only omits sending messages.
Further, in case of an isolation attack, forwarding decisions is generally preferable to forwarding proposals as it avoids additional latency induced by the subsequent protocol steps.
 
\subsubsection*{Optimizing Reads in a Hybrid Fault Model}
Troxy~\cite{li2018troxy} employs trusted execution environments (TEEs) to shift client functionality (like output consolidation) to the replica side, to grant clients transparent access to the BFT system. It also features a novel read optimization, in which Troxy actively manages a \textit{fast-read cache}, which reflects the state changes of the latest write, guaranteeing strong consistency. Functionality inside the TEE is assumed to never display Byzantine behavior and replicas need to be equipped with a TEE.
 
\subsubsection*{SMR Robustness}
Aardvark~\cite{clement2009making} improves the resilience of BFT systems against performance degradation attacks by introducing a set of design principles such as signed client requests, resource isolation, and regular view changes.
Spinning~\cite{veronese2009spin} proposes to let replicas alternately become the leader.
Prime~\cite{amir2011prime} enhances BFT systems with novel mechanisms against denial of service, corrupt message authentication code, and Pre-Prepare-Delay attacks.
RBFT~\cite{aublin2013rbft} proposes to redundantly run several instances of the ordering protocol, so that a misbehaving leader (that appears noticeable slow) can be detected and replaced.
None of these robustness enhancements seems to completely avoid the attack we presented.
 
 

\section{Conclusion}
\label{conclusion}

In this paper, we first explained an existing vulnerability in a commonly used BFT SMR design that is employed in PBFT~\cite{castro1999practical} and BFT-SMaRt~\cite{bessani2014state}.
This vulnerability is a consequence of the incorporation of an optimization that aims to avoid the costs of consensus when executing read-only operations.
We recalled that to preserve linearizability when using this optimization, clients need to wait for a larger quorum of replies. 
Further, we presented an attack in which a Byzantine leader isolates up to $f$ correct replicas, and thus can prevent clients from obtaining such a quorum of replies, hence breaking the liveness of the system.  

Subsequently, we showed two distinct, generic solutions that easily fit into both PBFT and BFT-SMaRt:  \textit{decision broadcasting} and \textit{decision forwarding}.
These solutions ensure execution in all correct replicas and enable correct system behavior when the read-only optimization is used: we can enable fast reads while preserving both linearizability and liveness properties of the system.

A further takeaway of this paper is that if the read-only optimization is used, state transfer needs to make sure the transferred state includes the last reply sent to each client to ensure recovered replicas can reply to outstanding requests issued by any client.

Moreover, we evaluated the advantage of the read-only optimization as well as our implemented variants, decision broadcasting and decision forwarding, in a LAN and in a WAN setup.
Evaluation results indicate that reads are  beneficial for system performance.
In a WAN environment, optimized reads are also substantially faster than their non-optimized counterparts: non-optimized read latencies are higher by a factor of $2.16\times$ to $2.54\times$.
When comparing decision broadcasting with decision forwarding, we observed that the forwarding technique achieves better performance due to its small overhead and is the suitable choice for replicas that are deployed close to each other (e.g., in a data center).
In a WAN environment, decision broadcasting achieves better latencies than forwarding during an isolation attack conducted by a Byzantine leader.

\bibliographystyle{IEEEtran}
\bibliography{IEEEabrv,bibliography}

\appendix

In this appendix we present a correctness proof for the two proposed modifications for dealing with the isolation attack. 
Our proofs consider specifically BFT-SMaRt~\cite{bessani2014state}, but can be easily adapted for PBFT~\cite{castro1999practical}.

\subsection{Correctness of Decision Broadcasting} 

\begin{theorem}
The decision broadcasting protocol modification preserves the following \textit{safety} property of the protocol: if a decision is made, then it is the same in all correct replicas.
\end{theorem}

\begin{proof}
This is because a decision is either made by the normal execution pattern in which a quorum\footnote{This is a generalization of the quorum size for arbitrary $f$ and $n$ as employed in BFT-SMaRt. In a typical $n=3f+1$ system, this quorum size is equal to $2f+1$.} $Q$ of $\lceil \frac{n+f+1}{2} \rceil$ received \texttt{ACCEPT}s are used for some value $v$, or, by validating a forwarded decided value $v'$ using an externally verifiable proof $\Gamma$, which contains a quorum $Q'$ of signed \texttt{ACCEPT} messages from replicas supporting $v'$. Since $\vert Q \cap Q'  \vert \geq f+1 $ holds, assuming that $v \neq v' $ would imply at least one \textit{correct} replica supported $v$ in $Q$ and $v'$ in $Q'$, which is a contradiction because correct replicas never equivocate, so $v'=v$ must hold.
\end{proof}

Before we continue with our liveness proof, we note that both PBFT and BFT-SMaRt guarantee that a request is eventually executed in at least $f+1$ correct replicas.
In cases in which a Byzantine leader reigns and less than $f+1$ correct replicas decide a value, we know that at least $f+1$ correct replicas stay undecided.
After a timeout, these replicas trigger a synchronization phase in BFT-SMaRt (or a view change in PBFT).
In the synchronization phase, replicas synchronize their logs and reach a similar state.
Note that this procedure might repeat (e.g., consider the possibility of $f$ cascading leader failures) until a stable leader takes over and ensures a decision is reached in at least $f+1$ correct replicas.

\begin{theorem}
The decision broadcasting protocol modification ensures the following \textit{liveness} property: a request issued by a correct client will be eventually executed by all correct replicas.
\end{theorem}

\begin{proof}
We first show that if a value is decided in some correct replica, then it is eventually decided in \textit{all} correct replicas.
For the start of our proof we assume that the value $v$ was decided in some correct replica $r$ for consensus instance $c$. Generally, a decision can be made by either gathering enough \texttt{ACCEPT} messages for a prepared proposal $v$ \textit{or} by receiving a $\langle \texttt{FWD-DECISION}, c, v, \Gamma \rangle$ message with valid proof $\Gamma$.
In both of these cases, upon deciding $v$ for consensus instance $c$, replica $r$ broadcasts the $\langle \texttt{FWD-DECISION}, c, v, \Gamma \rangle$ message to all other replicas, thus enabling all correct replicas to eventually decide $v$ (either during normal case operation pattern or after receiving this forwarded decision).
In particular, the fact that any correct replica that utilizes a $\langle \texttt{FWD-DECISION}, c, v, \Gamma \rangle$ message to decide will echo this message to all other replicas preserves the \textit{Termination} property of the consensus primitive.

Until now, we only showed that a decision made in a single correct replica eventually propagates to all correct replicas.
To show that an operation $o$ that needs to be ordered is eventually executed in \textit{all} correct replicas, we suggest to use the above result in combination with the liveness proof of the Mod-SMaRt algorithm in \cite{sousa2012byzantine} (Appendix, Theorem 2, using Lemmata 1-3). 
The Mod-SMaRt algorithm has specific behavioral properties that (i)~an operation received in a single correct replica is eventually received in all correct replicas, (ii)~eventually (after at most $f$ regency changes after \textit{GST}), a leader successfully drives progress in the protocol, and (iii)~if a single correct replica starts consensus, eventually consensus starts in at least $n-f$ replicas.

Using these properties, Sousa and Bessani~\cite{sousa2012byzantine} proved that an operation is eventually decided and executed in at least $f+1$ correct replicas.
This result combined with the decision broadcasting modification ensures that an operation executed by a correct replica is executed by all correct replicas.
\end{proof}

\subsection{Correctness of Decision Forwarding} 

\begin{theorem}
The decision forwarding protocol modification preserves the following \textit{safety} property of the protocol: if a decision is made, then it is the same in all correct replicas.
\end{theorem}

\begin{proof}
Analogous to the proof for decision broadcasting.
\end{proof}

\begin{theorem}
The decision forwarding protocol modification ensures the following \textit{liveness} property: a request issued by a correct client will be eventually executed by all correct replicas.
\end{theorem}

\begin{proof}
We will show that if a value $v$ is decided in \textit{at least one} correct replica, then it is eventually decided in \textit{all} correct replicas.
This result can be combined with the liveness proof of the Mod-SMaRt algorithm in \cite{sousa2012byzantine} (Appendix, Theorem 2), which shows any operation $o$ that needs to be ordered is eventually executed in at least $f+1$ correct replicas.

For the start of our proof we assume that the value $v$ was decided in some correct replica $r$ in consensus instance $c$. When using the decision forwarding extension, a decision can be made by either gathering enough \texttt{ACCEPT} messages for a prepared proposal $v$ \textit{or} by receiving a $\langle \texttt{FWD-DECISION}, c, v, \Gamma \rangle$ message with valid proof $\Gamma$, so we need to consider both cases.

\textit{(Case 1) Forwarded Decision:} The decision $v$ was made in $r$ by using a received, valid $\langle \texttt{FWD-DECISION}, c, v, \Gamma \rangle$ message. 
In this case, since replica $r$ decided $v$ using a forwarded decision, it will broadcast $\langle \texttt{FWD-DECISION}, c, v, \Gamma \rangle$ to all other replicas, which implies all correct replicas can eventually decide $v$.

\textit{(Case 2) Normal Operation Pattern:} The decision $v$ was made in $r$ by following the normal operation pattern.
In this case, we need to first consider how many correct replicas received a \texttt{PROPOSE} for $v$ in consensus instance $c$:
    
    \textit{(Case 2.a)} A \texttt{PROPOSE} for $v$ is received by \textit{all} correct replicas: 
   the concrete behavior depends on the order of messages received:
    
    \textit{(Case 2.a.i)}
    If the proposal is delayed, and some correct replica $r'$ receives $f+1$ \texttt{ACCEPT} messages before it receives the leader's proposal, then it sends a $\langle \texttt{REQ-DECISION}, c \rangle$ message to $2f$ other replicas. If $r'$ receives a valid $\langle \texttt{FWD-DECISION}, c, v, \Gamma \rangle$ message before completing the normal case operation pattern and uses the forwarded decision to decide, then it broadcasts the $\langle \texttt{FWD-DECISION}, c, v, \Gamma \rangle$ to all, thus enabling all correct replicas to eventually decide.
    Note that \textit{asking} for a decision does not block $r'$ from further participating in the agreement protocol.
        \balance
        
    \textit{(Case 2.a.ii)} 
    If no correct replica sent a forwarding request, it means all correct replicas executed the three-step agreement protocol for $c$ and decided.


    \textit{(Case 2.b)} There is a correct replica $r'$ that did not receive a \texttt{PROPOSE} for $v$ in $c$.
    Here, we need to consider how many correct replicas decided during the normal case operation pattern:
    
    \textit{(Case 2.b.i)} If at least $f+1$ correct replicas decide, $r'$ is guaranteed to have received $f+1$ \texttt{ACCEPTS} from these replicas and thus it sends a $\langle \texttt{REQ-DECISION}, c \rangle$ message to $2f$ other replicas.
    Since the set of $f+1$ correct and decided replicas intersects with any set of $2f+1$ probed replicas (the requester plus the $2f$ requested replicas) in at least one correct and decided replica, this replica will forward its decision and a valid $\langle \texttt{FWD-DECISION}, c, v, \Gamma \rangle$ will be received by $r'$. Then, $r'$ decides $v$ and broadcasts the $\langle \texttt{FWD-DECISION}, c, v, \Gamma \rangle$ to all replicas, ensuring all correct replicas eventually decide $v$.

     \textit{(Case 2.b.ii)}
     If less than $f+1$ correct replicas decided during the normal protocol execution, any correct, isolated replica may send a $\langle \texttt{REQ-DECISION}, c \rangle$ message to others, depending on whether it received $f+1$ matching \texttt{ACCEPTs} for $v$ (which depends now on Byzantine replicas sending such \texttt{ACCEPTs}). If \textit{any single} correct replica sends a $\langle \texttt{REQ-DECISION}, c \rangle$ message \textit{and} receives a valid $\langle \texttt{FWD-DECISION}, c, v, \Gamma \rangle$ response, it decides $v$, and broadcast this decision to all, thus enabling all correct replicas to eventually decide.
     If this is not the case, then we know at least $f+1$ replicas remain undecided, and will send  \texttt{STOP} messages, thus triggering a \textit{synchronisation phase} in which replicas synchronise their logs and continue under a new leader. After \textit{GST}, a stable leader will eventually reign after at most $f$ such regency changes. \end{proof}

\newpage

\end{document}